\documentclass{lmcs}
\title{Completeness of Relational Algebra via Cylindric Algebra}
\author{Jan Laštovička \lmcsorcid{0000-0002-0403-9709}}

\address{Department of Computer Science, Faculty of Science,
Palacky University Olomouc,
Czech Republic}	
\email{jan.lastovicka@upol.cz}

\def\tu#1{( #1)}
\def\C{\textrm{C}}
\def\S{\textrm{S}}
\def\D{\textrm{D}}
\def\A{\textrm{A}}
\def\R{\mathcal R}
\def\T{\mathcal T}
\def\X{\mathcal X}
\def\F{\mathcal F}
\def\V{\mathcal V}
\def\norm{norm}
\def\expr{expr}
\def\expri{\expr_\textrm{i}}
\def\eq{eq}
\def\coeq{coeq}
\def\al{align}
\def\id{\textrm{id}}

\def\Eq{{\textrm{Eq}}}
\def\FV{\textrm{FV}}
\def\cl{\textrm{cl}}
\def\DEE{\textrm{DEE}}

\def\gen{gen}
\def\cogen{cogen}
\def\genz{\gen_0}
\def\cogenz{\cogen_0}

\newcommand{\st}{%
	\ensuremath{\;|\;}
}
\begin{document}

\begin{abstract}
An alternative proof of the completeness of relational algebra with respect to allowed formulas of first-order logic is presented. The proof relies on the well-known embedding of relational algebra into cylindric algebra, which makes it possible to establish completeness in a more algebraic way. Building on this proof, we present an alternative algorithm that produces a relational expression equivalent to a given allowed formula. The main motivation for the present work is to establish a proof of completeness suitable for generalisation to relational models handling incomplete or vague information.
\end{abstract}

\maketitle
\section{Introduction}
Expressions of relational algebra \cite{Codd1970} form a procedural language suitable for querying relational databases. In the search for a declarative querying language, a natural approach is to consider all first-order formulas. For practical purposes, however, a declarative query must be translated into an equivalent procedural query that can be efficiently executed by a computer. This is not possible for every formula; for instance, the formula $\neg r(x)$ has no equivalent relational expression. In the following, we restrict our attention to allowed formulas \cite{Topor1987, GelderTopoer:safety}, which possess the desired property and are defined by a simple, easily verifiable structural condition.

Here, by the completeness of relational algebra, we refer to the theorem asserting that each allowed formula has a semantically equivalent relational expression, and conversely, each relational expression has a semantically equivalent allowed formula. As explained, we further require that for each allowed formula, a semantically equivalent relational expression can be constructed, which can then be efficiently evaluated. In \cite{GelderTopoer:safety}, the completeness theorem was proved, and an algorithm for converting allowed formulas into relational expressions was presented.

Cylindric algebras  \cite{HMT1971,HMTAN1981} were introduced by Tarski for the algebraization of first-order logic. Imieliński and Lipski \cite{ImLi1982} showed how a relational algebra can be naturally embedded into a cylindric algebra. By leveraging this embedding and partially following the approach of the original proof of completeness, we derive an alternative proof of the completeness theorem. 

The primary advantage of our alternative proof is that it is carried out within an appropriate cylindric algebra. This makes the proof technique applicable to the generalized relational models as well. For example, it is desirable to establish analogous completeness results for the relational model with incomplete information \cite{La2025} or for models incorporating similarity of values in domains \cite{BeVy:Rsbmd1}. To this end, cylindric algebras can be generalized to accommodate these cases, and corresponding lemmas and theorems can then be established, as shown in the following sections. 

Compared with \cite{GelderTopoer:safety}, our approach benefits from treating negation as a full-fledged logical connective, which makes it possible to define allowed formulas more clearly. Another advantage is that we permit equality in formulas, so there is no need to eliminate it before processing. Our algorithm for converting an allowed formula into an equivalent relational expression works in two steps. First, the allowed formula is transformed into a normalized formula, after which it is straightforwardly converted into a relational expression. The first normalization step preserves the structure of the input formula, a feature that can provide practical benefits.

The article has the following structure. In the second and third sections, the necessary notions of first-order logic and relational algebra are presented. Semantic equivalence of relational expressions with formulas is studied in the fourth section. Here also normalized formulas and they translation to relational expressions are defined. The section ends with remarks on domain-independent formulas. The fifth section is dedicated to allowed formulas. Finally, the algorithm for  normalization of allowed formulas and its correctness is given in the last sixth section.

\section{First-order logic}
We begin by defining the syntax and semantic of first-order logic, selecting from the available definitions those most appropriate for our purposes.

\paragraph{Syntax.} A language of the first-order logic consists of a finite set $\R$ of \emph{relation symbols}, each $r\in \R$ with a non-negative integer $\delta(r)$ called the \emph{arity} of $r$; \emph{tautology} $1$, \emph{equality}  $\approx$; a  finite set $\X$ of \emph{variables}; \emph{logical connectives} $\neg$, $\wedge$, $\vee$; \emph{existential quantifier} $\exists$; and auxiliary symbols $($,$)$ (left and right parenthesis). Note that, contrary to common practice, we do not consider equality as a relation symbol. Another unusual feature is that the set of variables is finite. Each language is determined by $\R$, $\delta$, and $\X$. The triple $\tu{\R,\delta, \X}$ is called a \emph{type}.

A \emph{formula} of a language of type $\tu{\R,\delta, \X}$ is inductively defined as follows:
\begin{enumerate}
\item $1$ is a formula;
\item if $r\in \R$ is a relation symbol of arity $n$, i.e., $\delta(r)=n$, and $x_1,\ldots,x_n\in\X$ are variables, then $r(x_1,\ldots,x_n)$ is a formula;
\item if $x_1$ and $x_2$ are variables, then $(x_1\approx x_2)$ is a formula;
\item if $\varphi$ is a formula, then $\neg\varphi$ is a formula;
\item if $\varphi$ and $\psi$ are formulas, then both $(\varphi\wedge\psi)$ and $(\varphi\vee\psi)$ are formulas;
\item if $\varphi$ is a formula and $x$ a variable, then $(\exists x)\varphi$ is a formula.
\end{enumerate}
Formulas given by the first three rules are called \emph{atomic formulas}, while those given by the remaining rules are called \emph{compound formulas}. Formulas constructed by the second rule are called \emph{relation formulas}, and those constructed by the third rule are called \emph{equality formulas}.

An expression $(\varphi\Rightarrow\psi)$ is an abbreviation for $(\neg\varphi\vee\psi)$, $(\varphi\Leftrightarrow\psi)$ is an abbreviation for $((\varphi\Rightarrow\psi)\wedge(\psi\Rightarrow\varphi))$, $(\forall x)\varphi$ is an abbreviation for $\neg(\exists x)\neg\varphi$, $(\varphi_1\wedge\varphi_2\wedge\cdots\wedge\varphi_n)$ is an abbreviation for $((\varphi_1\wedge\varphi_2)\wedge\cdots\wedge\varphi_n)$, similarly $(\varphi_1\vee\varphi_2\vee\cdots\vee\varphi_n)$ is an abbreviation for $((\varphi_1\vee\varphi_2)\vee\cdots\vee\varphi_n)$. We usually do not include the outermost parenthesis in formulas.

We assume that variables in $\X$ are linearly oredered. This order is used --- as the next definition demonstrates --- whenever we need to consider  variables in a specific order.  For $\{x_1,\ldots,x_n\}=X\subseteq \mathcal X$, $(\exists X)\varphi$ is an abbreviation for $(\exists x_1)\cdots(\exists x_n)\varphi$ and $(\forall X)\varphi$ is an abbreviation for $(\forall x_1)\cdots(\forall x_n)\varphi$. 

The set $\FV(\varphi)$ of \emph{free variables} of a formula $\varphi$ is defined as follows: $\FV(1)=\emptyset$; $\FV(r(x_1,\ldots,x_n))=\{x_1,\ldots,x_n\}$; $\FV(x_1\approx x_2)=\{x_1,x_2\}$;  $\FV(\neg\varphi)=\FV(\varphi)$; $\FV(\varphi\vee\psi)=\FV(\varphi\wedge\psi)=\FV(\varphi)\cup\FV(\psi)$; $\FV((\exists x)\varphi)=\FV(\varphi)-\{x\}$.

\paragraph{Semantics.} A \emph{structure} of a language of type $\tu{\R,\delta, \X}$ is a triple $\tu{M,\approx^M, \R^M}$ consisting of a non-empty finite set $M$, called \emph{domain},  the equality relation $\approx^M$ on $M$, i.e., $\approx^M=\{\tu{m,m}\st m\in M\}$, and an indexed set $\R^M=\{(r^M)_{r\in \R}\st r^M\subseteq M^{\delta(r)}\}$ 
 of relations. Since our structures model relational databases, we restrict our attention to finite domains. A given structure $\tu{M,\approx^M,\R^M}$ is also denoted simply by $M$.  For the remainder of the paper, we fix a language of type $\tu{\R,\delta, \X}$.
 
 Let $M$ be a structure.  A mapping $v:\!\X\to M$ that assigns to each variable $x$ an element $v(x)\in M$ is called an \emph{$M$-valuation}. When $M$ is specified by the context, $v$ is simply called  a \emph{valuation}.  For a variable $x$ and valuations $u$ and $v$, we write $u=_xv$ if $u(y)=v(y)$ for each variable $y$ distinct from $x$. For a set of variables $X=\{x_1,\ldots,x_n\}$, we write $u=_Xv$ if there exists valuations $v_1,\ldots, v_{n+1}$ such that $v_i=_{x_i}v_{i+1}$ for each $1\leq i\leq n$, $v_1=v$, and $v_{n+1}=u$.
 
We do not define when a formula is true in a structure with respect to a valuation. Rather,  for a given formula and structure, we define the set of all valuations for which the formula is true. We begin with a few operations that will be useful for this purpose. 
 
 When $M$ is specified by the context, the set $M^\X$ of all $M$-valuations is denoted by $\V$. For a set $V\subseteq\V$, we denote its complement relative to $\V$ by $\overline V$, i.e. $\overline{V}=\V\setminus V$.   
For a variable $x$, we define operator $\C_x\!:\V\to \V$ called \emph{cylindrification} by
\begin{align}
\C_x(V)=\{v\in\V\st \text{there is }u\in V\text{ such that } u=_xv\}.
\end{align}
Clearly, cilindrifications are closure operators.

\begin{lem}\label{lem:Cx_overline}
The following holds.
\begin{enumerate}
\item $\overline{\C_x(V)}\subseteq\overline{V}$,
\item $\C_x(\overline{\C_x(V)})= \overline{\C_x(V)}$,
\item $\overline{\C_x(V)}\subseteq \C_x(\overline V)$.
\end{enumerate}
\end{lem}
\begin{proof} ~\newline
\begin{enumerate}
\item This follows from $V\subseteq\C_x(V)$.
\item Clearly, $\C_x(\overline{\C_x(V)})\supseteq\overline{\C_x(V)}$. Let $v\in\C_x(\overline{\C_x(V)})$. Then there is $v'\in\V$ such that $v'\in\overline{\C_x(V)}$ and $v'=_xv$. For each $v''\in V$, we have $v''\not=_xv'$. Let $v''\in V$ be such that $v''\not=_xv'$. Then there is a variable $y\in\X$ distinct from $x$ with $v''(y)\not=v'(y)$.  But since $v'(y)=v(y)$, we have $v''\not=_xv$. Hence, $v\in\overline{\C_x(V)}$.
\item From (1) and (2), it follows that $\overline{\C_x(V)}=\C_x(\overline{\C_x(V)})\subseteq\C_x(\overline{V})$.  \qedhere
\end{enumerate}

\end{proof}

For variables $x,y\in \X$, we define set $\D_{xy}\subseteq\V$ called \emph{diagonal} by
\begin{align}
\D_{xy}=\{v\in \V\st v(x)=v(y)\}. 
\end{align}

The relation $||\varphi||_M$, called the \emph{value} of $\varphi$ in $M$, is defined as follows:
\begin{enumerate}
\item $||1||_M=\V$;
\item $||r(x_1,\ldots,x_n)||_M=\{v\in\V\st \tu{v(x_1),\ldots,v(x_n)}\in r^M \}$;
\item $||x\approx y||_M=\D_{xy}$;
\item if $\varphi$ is a formula, then $||\neg\varphi||_M=\overline{||\varphi||_M}$;
\item if $\varphi$ and $\psi$ are formulas, then
\begin{align*}
||\varphi\wedge\psi||_M&= ||\varphi||_M\cap||\psi||_M,\\
||\varphi\vee\psi||_M&= ||\varphi||_M\cup||\psi||_M;
\end{align*}
\item if $\varphi$ is a formula and $x$ is a variable, then
\begin{align*}
||(\exists x)\varphi||_M=C_x(||\varphi||_M).
\end{align*}
\end{enumerate}
 
If $M$ is clear from the context, we write $||\varphi||$ instead of $||\varphi||_M$. It is easy to see that $\C_x(||\varphi||)=||\varphi||$, whenever $x\notin\FV(\varphi)$.

For a set $X=\{x_1,\ldots, x_n\}$, we define an unary operation $\C_S$ on $\V$ by
\begin{align}
\C_S(V)=\C_{x_1}(\cdots \C_{x_n}(V)\cdots). 
\end{align}
Clearly, it holds $||(\exists X)\varphi||=\C_X(||\varphi||)$.

Although it is not essential for the following discussion, it is interesting to note that the semantics of first-order logic can be given by a specific algebra. \emph{Cylindric algebra of dimension $I$} is an algebra structure
 $$\tu{A,\vee,\wedge,',0,1,c_i,d_{ij}}, \quad i,j\in I,$$
 such that $0,1$, and $d_{ij}$, for $i,j\in I$, are elements of $A$; $'$ and $c_i$, for $i\in I$, are unary operations; $\vee$  and $\wedge$ are binary operations; and the following conditions are satisfied for any $a,b\in A$ and any $i,j,k\in I$: 
\begin{enumerate}
 \item $\tu{A,\vee,\wedge,',0, 1}$ is a Boolean algebra,
 \item $c_i(0)=0$,
 \item $a\wedge c_i(a)=a$,
 \item $c_i(a\wedge c_i(b))=c_i(a)\wedge c_i(b)$,
 \item $c_i(c_j(a))=c_j(c_i(a))$,
 \item $d_{ii}=1$,
 \item if $i\neq j$ and $i\neq k$, then $d_{jk}=c_i(d_{ji}\wedge d_{ik})$,
 \item if $i\neq j$, then $c_i(d_{ij}\wedge a)\wedge c_i(d_{ij}\wedge a')=0$.
\end{enumerate}
It is easy to verify that 
\begin{align}
\tu{2^{\V},\cup,\cap,\overline{\,\cdot\,},\emptyset,\V,\C_x,\D_{xy}}, \quad x,y\in \X 
\end{align}
is a cylindric algebra of dimension $\X$.

\section{Relational algebra}
Let $M$ be a structure. \emph{Relation schemes} are subsets of $\X$. Let $S$ be a relation scheme. A \emph{tuple $t$ over $S$} is a mapping $t\!:S\to M$, i.e., $t\in M^S$. For a tuple $t$ over $S$ and a relation scheme $S'\subseteq S$, the restriction $t|_{S'}=\{\tu{y,t(y)}\st y\in S'\}$ of $t$ to $S'$ is a tuple over $S'$. A \emph{relation $T$ over $S$} is any set of tuples over $S$, i.e., $T\subseteq M^S$. Note that, since $M$ and $S$ are finite sets, $T$ is also finite. 

Let $T_1$ and $T_2$ be relations over $S_1$ and $S_2$, respectively. Then the \emph{join} of $T_1$ and $T_2$ is the relation $T_1\bowtie T_2$ over $S_1\cup S_2$ defined by
\begin{align}
T_1\bowtie T_2=\{t\in M^{S_1\cup S_2}\st t|_{S_1}\in T_1\textrm{ and } t|_{S_2}\in T_2\}.
\end{align}
Let $T$ be a relation over $S$, and let $S'$ be a relation scheme that is a subset of $S$. Then the \emph{projection} of $T$ onto $S'$ is the relation $\pi_{S'}(T)$ over $S'$ defined by
\begin{align}
\pi_{S'}(T)=\{t|_{S'}\st t\in T\}. 
\end{align}
Let $T$ be a relation over $S$, and let $x_1$ and $x_2$ be variables from the relation scheme $S$. Then the \emph{restriction} of $T$ by the equality of $x_1$ and $s_2$ is the relation $\sigma_{x_1\approx x_2}(T)$ over $S$ defined by
\begin{align*}
\sigma_{x_1\approx x_2}(T)=\{t\in T\st t(x_1)=t(x_2)\}.
\end{align*}
Let $T$ be a relation over $S$, and let $x$ and $y$ be variables such that $x\in S$ and $y\notin S$. Then the \emph{renaming} in $T$ variable $x$ by $y$ is the relation $\rho_{y\leftarrow x}(T)$ over $S'=(S\setminus\{x\})\cup\{y\}$ defined by
\begin{align*}
\rho_{y\leftarrow x}(T)=\{\rho_{y\leftarrow x}(t)\st t\in T\}, 
\end{align*}
where $\rho_{y\leftarrow x}(t)$ is tuple over $S'$ given by
\[
\rho_{y\leftarrow x}(t)(z) =
\begin{cases} 
t(x)\quad&\textrm{ if }z=y,\\
t(z) &\textrm{ otherwise}.\\
\end{cases}
\]

In their 1984 article \cite{ImLi1982}, Imieliński and Lipski establish a close connection between relational algebras and cylindric algebras for structures. We begin by defining a natural mapping that assigns a set of valuations to a relation. If a structure $M$ is given by the context, we denote by $\T$ the set of all relations over arbitrary schemes; that is, $\T=\{T\subseteq M^S\st S\subseteq \X\}$. We define a mapping $\epsilon\!:\T\to\V$ given by
\begin{align}
\epsilon(T)=\{v\in\V\st v|_S\in T\}. 
\end{align}
The following theorem, taken from \cite{ImLi1982}, establishes the connection.
\begin{thm}\label{thm:imli}
Let $T$, $T_1$, and $T_2$ be relations over $S$, $S_1$, and $S_2$, respectively. Then the following hold:
\begin{enumerate}
\item  $\epsilon(T_1\cup T_2)=\epsilon(T_1)\cup \epsilon(T_2)$, if $S_1=S_2$;
\item  $\epsilon(T_1\setminus T_2)=\epsilon(T_1)\cap \overline{\epsilon(T_2)}$, if $S_1=S_2$;
\item $\epsilon(T_1\bowtie T_2)=\epsilon(T_1)\cap \epsilon(T_2)$;
\item $\epsilon(\pi_{S'}(T))=\C_{S\setminus S'}(\epsilon(T))$, where $S'\subseteq S$;
\item $\epsilon(\sigma_{x\approx y}(T))=\epsilon(T)\cap\D_{xy}$, where $x,y\in S$;
\item $\epsilon(\rho_{y\leftarrow x}(T))=\C_x(\epsilon(T)\cap\D_{xy})$.
\end{enumerate}
 \end{thm}

We define syntax and semantics of relational expressions.
\paragraph{Syntax.} To align more closely with modern relational databases, we model the scheme of a database by a subset of relation formulas. A \emph{database scheme} of a language of a type $\tu{\R,\delta,\X}$ is a tuple $\tu{\R,\delta,\X,\F}$, where $\F$ is a set of atomic formulas such that, for each $n$-ary relation symbol $r\in \R$, there is exactly one atomic formula $r(x_1,\ldots,x_n)\in\F$, where the variables $x_1,\ldots,x_n$ are pairwise distinct.  

A \emph{relational expression} over a relation scheme for a database scheme  $\tu{\R,\delta,\X,\F}$ is defined inductively  as follows:
\begin{enumerate}
\item $\DEE$ is a relational expression over $\emptyset$;
\item if $r(x_1,\ldots,x_n)\in\mathcal\F$, then $r$ is a relational expression over $\{x_1,\ldots, x_n\}$;
\item if $E_1$ and $E_2$ are relational expression over $S$, then $E_1\cup E_2$ and $E_1-E_2$ are relational expression over $S$;
\item if $E_1$ and $E_2$ are relational expression over $S_1$ and $S_2$, respectively, then $R_1\bowtie R_2$ is a relational expression over $S_1\cup S_2$;
\item if $E$ is a relational expression over $S$ and $S'\subseteq S$, then $\pi_{S'}(E)$ is a relational expression over $S'$;
\item if $E$ is a relational expression over $S$ and $x,y\in S$, then $\sigma_{x\approx y}(E)$ is a relational expression over $S$;
\item if $E$ is a relational expression over $S$ and $x\in S$ and $y\in \X\setminus S$, then $\rho_{y\leftarrow x}(E)$ is a relational expression over $(S\setminus\{x\})\cup\{y\}$.
\end{enumerate}
The relation scheme of a relational expression $E$ is denoted by $\S(E)$.

\paragraph{Semantics.} Any structure of a language of type $\tu{\R,\delta, \X}$ can be called a \emph{database} over database scheme $\tu{\R,\delta,\X,\F}$.  We fix a database scheme $\tu{\R,\delta,\X,\F}$ for the rest of the paper. 

The \emph{value} of a relational expression $E$ over $S$ in a database $M$ is the relation  $||E||_M$ over $S$ defined as follows:
\begin{enumerate}
\item $||\DEE||_M=\{\emptyset\}$.
\item if $r$ is a relational expression, then
$$||r||_M=\{t\in M^S\st \tu{t(x_1),\ldots,t(x_n)}\in r^M \},$$
 where $x_1,\ldots x_n$ are variables such that $r(x_1,\ldots,x_n)\in\F$ and $S=\{x_1,\ldots,x_n\}$;
\item if $E_1$ and $E_2$ are relational expressions over $S$, then
\begin{align*}
||E_1\cup E_2||_M=||E_1||_M\cup||E_2||_M,\\
||E_1-E_2||_M= ||E_1||_M\setminus||E_2||_M;
\end{align*}
\item if $E_1$ and $E_2$ are relational expressions, then
\begin{align*}
||E_1\bowtie E_2||_M= ||E_1||_M\bowtie||E_2||_M;
\end{align*}
\item if $E$ is a relational expression over $S$ and $S'\subseteq S$, then 
\begin{align*}
||\pi_{S'}(E)||_M= \pi_{S'}(||E||_M);
\end{align*}
\item if $E$ is a relational expression over $S$ and $x,y\in S$, then 
\begin{align*}
||\sigma_{x_1\approx x_2}(E)||_M=\sigma_{x_1\approx x_2}(||E||_M);
\end{align*}
\item if $E$ is a relational expression over $S$, $x\in S$ and $y\in \X\setminus S$, then 
\begin{align*}
||\rho_{y\leftarrow x}(E)||_M=\rho_{y\leftarrow x}(||E||_M).
\end{align*}
\end{enumerate}

\section{Semantic equivalence}\label{sec:sem_eq}
We begin with the definitions of semantic equivalence for formulas and expressions.
Formulas $\varphi_1$ and $\varphi_2$ are \emph{semantically equivalent} if $||\varphi_1||_M=||\varphi_2||_M$ for every structure $M$. Expressions $E_1$ and $E_2$ are \emph{semantically equivalent} if $||E_1||_M=||E_2||_M$ for every database $M$. A relational expression $E$ is \emph{semantically equivalent} to a formula $\varphi$ if
\begin{align}
\epsilon(||E||_M)=||\varphi||_M 
\end{align}
for every database $M$.
\begin{thm}
Every relational expression  is semantically equivalent to a formula.  
\end{thm}
\begin{proof}
Follows directly from Theorem \ref{thm:imli} and semantics rules for formulas.
\end{proof}
On the other hand, there are formulas with no semantically equivalent relational expressions. For example, there is no relational expression that is semantically equivalent to a formula $\neg\varphi$, where $\varphi$ is any relation formula.

The clauses of Theorem~\ref{thm:imli} define a subset of formulas for which semantically equivalent relational expressions are obtained directly. \emph{Normalized formulas} are inductively defined by the following rules:
\begin{enumerate}
\item $1$ is normalized;
\item if $\varphi\in\F$, then $\varphi$ is normalized;
\item if $\varphi$ is normalized, $x\in\FV(\varphi)$, and $y$ is any variable, then $\varphi\wedge(x\approx y)$ is normalized;
\item if $\varphi$ and $\psi$ are normalized with $\FV(\varphi)=\FV(\psi)$, then $\varphi\wedge\neg\psi$ and $\varphi\vee\psi$ are normalized;
\item if $\varphi$ and $\psi$ are normalized, then $\varphi\wedge\psi$ is normalized;
\item if $\varphi$ is normalized and $x$ is a variable, then $(\exists x)\varphi$ is normalized.
\end{enumerate}
A normalized formula $\varphi$ can be easily translated to a semantically equivalent relational expression $\expr(\varphi)$ by the following rules:
\begin{enumerate}
\item $\expr(1)=\DEE$;
\item if $\varphi$ is of the form $r(x_1,\ldots,x_n)\in\F$, then $\expr(\varphi)=r$;
\item if $\varphi$ is of the form $\varphi_1\wedge(x\approx y)$ and $x,y\in\FV(\varphi_1)$, then $\expr(\varphi)=\sigma_{x\approx y}(\expr(\varphi_1))$;
\item if $\varphi$ is of the form $\varphi_1\wedge(x\approx y)$ and $x\in\FV(\varphi_1)$ and $y\notin\FV(\varphi_1)$, then $\expr(\varphi)=\sigma_{x\approx y}(E_1\times\rho_{y\leftarrow x}(\pi_{\{x\}}(E_1)))$, where $E_1=\expr(\varphi_1)$;
\item if $\varphi$ is of the form $\varphi_1\wedge\neg\varphi_2$, then $\expr(\varphi)=\expr(\varphi_1)-\expr(\varphi_2)$;
\item if $\varphi$ is of the form $\varphi_1\vee\varphi_2$, then $\expr(\varphi)=\expr(\varphi_1)\cup\expr(\varphi_2)$;
\item if $\varphi$ is of the form $\varphi_1\wedge\varphi_2$ and $\varphi_2$ is not of the form $\neg\varphi_3$, then $\expr(\varphi)=\expr(\varphi_1)\bowtie\expr(\varphi_2)$;
\item if $\varphi$ is of the form $(\exists x)(\varphi_1\wedge(x\approx y))$, then $\expr(\varphi)=\rho_{y\leftarrow x}(\expr(\varphi_1))$;
\item if $\varphi$ is of the form $(\exists x)\varphi_1$ and $\varphi_1$ is not of the form  $\varphi_2\wedge(x\approx y)$, then  $\expr(\varphi)=\pi_{\FV(\varphi_1)-\{x\}}(\expr(\varphi_1))$.
\end{enumerate}

\begin{lem}\label{lem:norm_form}
Each normalized formula $\varphi$ is semantically equivalent to the relational expression $\expr(\varphi)$. 
\end{lem}
\begin{proof}
A direct consequence of Theorem \ref{thm:imli}. 
\end{proof}

We show that every relation formula is semantically equivalent to a normalized formula. 
If $\varphi$ is a formula and $x,y$ are variables. Then by $(\varphi)[x/y]$ we denote the formula
\begin{align}
(\exists x)(\varphi\wedge(x\approx y)).
\end{align}
If $x_1,\ldots,x_n$ and $y_1,\ldots,y_n$ are variables, then by
\begin{align}
(\varphi)[x_1/y_1,\ldots,x_n/y_n] 
\end{align}
we denote the formula
\begin{align}
(\cdots(\varphi)[x_1/y_1]\cdots)[x_n/y_n] 
\end{align}

Let $r(x_1,\ldots,x_n)$ be a relation formula such that variables $x_1,\ldots,x_n$ are pairwise distinct, and let $y_1,\ldots,y_n$ be any variables. Then by
\begin{align}
r(x_1/y_1,\ldots,x_n/y_n), 
\end{align}
we denote the formula
\begin{align}
(r(x_1,\ldots,x_n))[x_1/z_1,\ldots,x_n/z_n,z_1/y_1,\ldots,z_n/y_n],
\end{align}
where $z_1,\ldots,z_n$ are pairwise distinct variables, such that sets $\{z_1,\ldots,z_n\}$ \\and $\{x_1,\ldots,x_n,y_1,\ldots,y_n\}$ are disjoint. We assume that $\X$ contains enough variables for choosing $z_1,\ldots,z_n$. For example, 
\begin{align*}
r(x_1/x_1,x_2/x_1)&=(r(x_1,x_2))[x_1/z_1,x_2/z_2,z_1/x_1,z_2/x_1]\\
&=(\exists z_2)((\exists z_1)((\exists x_2)((\exists x_1)(r(x_1,x_2)\wedge(x_1\approx z_1))\wedge(x_2\approx z_2))\\
&\phantom{===}\wedge(z_1\approx x_1))\wedge(z_2\approx x_1)). 
\end{align*}
The last formula can be simplified to the semantically equivalent formula $(\exists x_2)(r(x_1,x_2)\wedge(x_2\approx x_1))$. For the sake of readability, we will make similar simplifications in the examples given later in the text.

\begin{lem}\label{lem:rename}
The formulas  $r(x_1/y_1,\ldots,x_n/y_n)$ and $r(y_1,\ldots,y_n)$ are semantically equivalent. 
\end{lem}
\begin{proof}
We have  $||r(x_1/y_1,\ldots,x_n/y_n)||=$
\begin{align*}
\C_{\{z_1,\ldots,z_n\}}(&\C_{\{x_1,\ldots,x_n\}}(||r(x_1,\ldots,x_n)||\cap\D_{x_1z_1}\cap\cdots\cap\D_{x_nz_n}) \\
&\cap\D_{z_1y_1}\cap\cdots\cap\D_{z_ny_n}).
\end{align*}
Let $v\in\V$. We have $v\in||r(x_1/y_1,\ldots,x_n/y_n)||$ if and only if there exists $v'\in\V$ such that $v'=_{\{z_1,\ldots,z_n\}}v$ and $v'(z_i)=v'(y_i)=v(y_i)$ for each $1\leq i\leq n$ and there exists $v''\in\V$ such that $v''=_{\{x_1,\ldots,x_n\}}v'$ and $v'(z_i)=v''(z_i)=v''(x_i)$ for each $1\leq i\leq n$ and $\tu{v''(x_1),\ldots,v''(x_n)}\in r^M$. 

Suppose $v\in||r(x_1/y_1,\ldots,x_n/y_n)||$. Then we have $v(y_i)=v'(z_i)=v''(x_i)$ for each $1\leq i\leq n$ and $\tu{v(y_1),\ldots,v(y_n)}=\tu{v''(x_1),\ldots,v''(x_n)}\in r^M$. We showed that $v\in||r(y_1,\ldots,y_n)||$.

Suppose $v\in||r(y_1,\ldots,y_n)||$. Let $v'$ be a valuation such that
\begin{align*}
v'(x)=
\begin{cases}
   v(x) & \mbox{if } x\notin\{z_1,\ldots,z_n\},\\
    v(y_i) & \mbox{if } x=z_i \mbox{ for some }1\leq i\leq n,
\end{cases} 
\end{align*}
and let $v''$ be a valuation such that
\begin{align*}
v''(x)=
\begin{cases}
   v'(x) & \mbox{if } x\notin\{x_1,\ldots,x_n\},\\
    v'(z_i) & \mbox{if } x=x_i \mbox{ for some }1\leq i\leq n.
\end{cases} 
\end{align*} 
Clearly, $v=_{\{z_1,\ldots,z_n\}}v'$ and $v'=_{\{x_1,\ldots,x_n\}}v''$. We have $v''(x_i)=v'(z_i)=v(y_i)$ for each $1\leq i\leq n$.  Therefore $\tu{v''(x_1),\ldots,v''(x_n)}\in r^M$. We showed that $v\in||r(x_1/y_1,\ldots,x_n/y_n)||$.
\end{proof}

We recapitulate known results on formulas equivalent to relational expressions; 
for further details, see, e.g., \cite{book:alice}. Two structures $M_1$ and $M_2$ \emph{differ only in domains} if 
\begin{align}
r^{M_1}=r^{M_2} 
\end{align}
for every relation symbol $r$. A formula $\varphi$ is called \emph{domain independent}  if 
\begin{align}
\pi_{\FV(\varphi)}(||\varphi||_{M_1})=\pi_{\FV(\varphi)}(||\varphi||_{M_2}),
\end{align}
for every two structures $M_1$ and $M_2$ which differs only in domains.

As the following lemma shows, there is no point in defining domain independent relational expressions.
Two databases $M_1$ and $M_2$ \emph{differ only in domains} if their structures differ only in domains.
\begin{lem}\label{lem:expr_independent}
Let $E$ be a relational expression. Then
\begin{align}
||E||_{M_1}= ||E||_{M_2}
\end{align}
for every two databases $M_1$ and $M_2$ which differ only in domains.
\end{lem}
\begin{proof}
Follows directly from the definitions of relational operations.
\end{proof}

\begin{thm}\label{thm:dom_idep}
If a formula is semantically equivalent to a relational expression, then it is domain independent. 
\end{thm}
\begin{proof}
Let $\varphi$ be a formula such that there exists a relational expression $E$ semantically equivalent to $\varphi$. Let $M_1,M_2$ be databases which differ only in domains. Then, by Lemma \ref{lem:expr_independent}, we have $||\varphi||_{M_1}=\epsilon(||E||_{M_1})=\epsilon(||E||_{M_2})=||\varphi||_{M_2}$. 
\end{proof}

The active domain of a database is the set of all values occurring in the relations of database. More precisely, for database $M$, we denote by $\A(M)\subseteq M$ the set called \emph{active domain} of $M$, defined as 
\begin{align}
\A(M)=\bigcup_{r(x_1,\ldots,x_n)\in\F}\;\bigcup_{1\leq i\leq n}\{m_i\st (m_1,\ldots,m_n)\in r^M\}.
\end{align}

\begin{thm}
Every domain independent formula is semantically equivalent to a relational expression. 
\end{thm}
\begin{proof}
For every scheme $S$, a relational expression $E_S$ such that $||E_S||_M={\A(M)}^S$ for every database $M$ can be easily constructed. For a formula $\varphi$, we define the relational expression $\expri(\varphi)$ by the following rules.
\begin{enumerate}
\item $\expri(1)=\DEE$,
\item $\expri(r(x_1,\ldots,x_n))=\expr(r(x_1,\ldots,x_n))$,
\item $\expri(x_1\approx x_2)=\sigma_{x_1\approx x_2}(E_{\{x_1,x_2\}})$,
\item $\expri(\neg\varphi)=E_{\FV(\varphi)}\setminus\expri(\varphi)$,
\item $\expri(\varphi\wedge\psi)=\expri(\varphi)\bowtie\expri(\psi)$,
\item $\expri(\varphi\vee\psi)=\expri(\neg(\neg\varphi\wedge\neg\psi))$,
\item $\expri((\exists x)\varphi)=\pi_{\FV(\varphi)\setminus\{x\}}(\expri(\varphi))$.
\end{enumerate}

Let $\varphi$ be a domain independent formula and let $M$ be a database. Define $M_A$ to be the database obtained from $M$ by replacing its domain with $\A(M)$. Then $\epsilon(||\expri(\varphi)||_{M_A})=||\varphi||_{M_A}$, and hence $||\expri(\varphi)||_{M_A}=\pi_{\FV(\varphi)}(||\varphi||_{M_A})$. Since $\varphi$ is domain independent, we have $\pi_{\FV(\varphi)}(||\varphi||_{M_A})=\pi_{\FV(\varphi)}(||\varphi||_{M})$, and by Lemma \ref{lem:expr_independent},  $||\expri(\varphi)||_M=||\expri(\varphi)||_{M_A}$. Combining these equalities, we obtain  $||\expri(\varphi)||_M=\pi_{\FV(\varphi)}(||\varphi||_M)$ and therefore \\$\epsilon(||\expri(\varphi)||_M)=||\varphi||_M$. Thus, $\varphi$ is semantically equivalent to $\expri(\varphi)$.
\end{proof}
Note that the translation of domain-independent formulas, as given in the proof of the preceding theorem, is purely of theoretical interest. Domain-independent formulas have another major disadvantage. 
A formula  $\varphi$ is \emph{satisfiable} if there exists a structure $M$ such that $||\varphi||_M$ is nonempty. 
Trakhtenbrot \cite{Trakhtenbrot1950} showed that the problem of determining whether a formula $\varphi$ is satisfiable is undecidable. 
It follows that the problem of determining whether a given formula is domain independent is also undecidable, since the formula $\varphi\wedge\neg r(x)$ is domain independent if and only if $\varphi$ is not satisfiable, where the variable $x$ is not free in $\varphi$. 

For practical purposes, researchers propose sufficiently large, decidable subsets of the set of domain-independent formulas (for an overview,  see  \cite{Topor09}).

\section{Allowed formulas}

We start with the problem of determining whether, for a given formula $\varphi$ and variables $x,y$, we have $||\varphi||_M\subseteq||x\approx y||_M$ for every structure $M$. This problem is undecidable, since for a formula $\varphi$ and variables $x,y\notin\FV(\varphi)$, the formula $\varphi\wedge r(x,y)$ has the desired property if and only if $\varphi$ is not satisfiable. Therefore, we infer from the structure of a given formula only a sufficient condition for the property being studied. 

Let $M$ be a database. For a set $A$, we denote by $\id_A$ the identity relation on $A$, i.e., $\id_A=\{\tu{a,a}\st a\in A\}$. For a binary relation $R$ on $\X$, we denote by $\Eq(R)$ the smallest equivalence on $\X$ (with respect to the set inclusion) containing $R$.

By the following rules, for a formula $\varphi$, we define equivalences $\eq(\varphi)$ and $\coeq(\varphi)$ on $\X$, called \emph{equality} and \emph{coequality of variables} in $\varphi$, respectively:
\begin{enumerate}
\item $\eq(1)=\coeq(1)=\id_\X$,
\item $\eq(r(x_1,\ldots,x_n))=\id_\X$,
\item $\coeq(r(x_1,\ldots,x_n))=\id_\X$,
\item $\eq(x_1\approx x_2)=\Eq(\{\tu{x_1,x_2}\}),$ 
\item $\coeq(x_1\approx x_2)=\id_\X$,
\item $\eq(\neg\varphi)=\coeq(\varphi)$,
\item $\coeq(\neg\varphi)=\eq(\varphi)$,
\item $\eq(\varphi\wedge\psi)=\Eq(\eq(\varphi)\cup\eq(\psi))$,
\item $\coeq(\varphi\wedge\psi)=\coeq(\varphi)\cap\coeq(\psi)$,
\item $\eq((\exists x)\varphi)=\Eq(\eq(\varphi)|_{(\X\setminus\{x\})^2})$,
\item $\coeq((\exists x)\varphi)=\Eq(\coeq(\varphi)|_{(\X\setminus\{x\})^2})$.
\end{enumerate}
The rules just defined do not cover formulas in the form of a disjunction. A formula in the form $\varphi\vee\psi$ can be treated as the formula $\neg(\neg\varphi\wedge\neg\psi)$. This rule for disjunction also applies in the following recursive definitions, unless disjunction is treated explicitly. For example, we have $\eq((x_1\approx x_2)\wedge r(x_1,x_3)\wedge(x_2\approx x_3))=\Eq(\{\tu{x_1, x_2},\tu{x_2,x_3}\})$.

The meaning of equality and coequality of variables in formulas are captured by the following lemma.
\begin{lem}\label{lem:eq_coeq}
If $\tu{x,y}\in\eq(\varphi)$, then $||\varphi||\subseteq\D_{xy}$; and if $\tu{x,y}\in\coeq(\varphi)$, then $\overline{||\varphi||}\subseteq\D_{xy}$.
\end{lem}
\begin{proof}
Since $\D_{xx}=\V$ for any variable $x$, the assertion trivially holds for the case when $x=y$. So we can suppose that $x\neq y$.  We prove this assertion by structural induction on formulas. 

Clearly, the assertion holds for atomic formulas. Assume that the assertion holds for formulas $\varphi$ and $\psi$.

If $\tu{x,y}\in\eq(\neg\varphi)$, then $\tu{x,y}\in\coeq(\varphi)$, and by assumption, $\overline{||\varphi||}\subseteq\D_{xy}$. Thus, $||\neg\varphi||\subseteq\D_{xy}$. Similarly, one can prove the case when $\tu{x,y}\in\coeq(\neg\varphi)$. Therefore, the assertion holds for negated formulas.

Let $\tu{x,y}\in\eq(\varphi\wedge\psi)$. Then $\tu{x,y}\in\Eq(\eq(\varphi)\cup\eq(\psi))$. There are pairwise distinct variables $x_1,\ldots,x_n$ such that $x_1=x$, $x_n=y$, and $\tu{x_i,x_{i+1}}\in\eq(\varphi)\cup\eq(\psi)$ for every $1\leq i< n$. Thus, for every $1\leq i<n$, we have $||\varphi||\subseteq\D_{x_ix_{i+1}}$ or $||\psi||\subseteq\D_{x_ix_{i+1}}$; moreover, $||\varphi||\cap||\psi||\subseteq\D_{x_ix_{i+1}}$.  Using 7. axiom of Cylindric algebras, we obtain
\begin{align*}
||\varphi\wedge\psi||&= ||\varphi||\cap||\psi||\subseteq\D_{x_1x_2}\cap\cdots\cap\D_{x_{n-1}x_n}\\
&\subseteq\C_{\{x_2,\ldots,x_{n-1}\}}(\D_{x_1x_2}\cap\cdots\cap\D_{x_{n-1}x_n})=\D_{xy}.
\end{align*}
If $\tu{x,y}\in\coeq(\varphi\wedge\psi)$, then $\tu{x,y}\in\coeq(\varphi)$ and $\tu{x,y}\in\coeq(\psi)$. By assumption, $||\varphi||\subseteq\D_{xy}$ and  $||\psi||\subseteq\D_{xy}$. We have $\overline{||\varphi\wedge\psi||}=\overline{||\varphi||}\cup\overline{||\psi||}\subseteq \D_{xy}$. We proved the assertion for formulas in the form of a conjunction.

Let $\tu{x,y}\in\eq((\exists z)\varphi)$. Since $x\neq y$, $\tu{x,y}\in\eq(\varphi)$. We have
\begin{align*}
||(\exists z)\varphi||=\C_z(||\varphi||)\subseteq\C_z(\D_{xy})=\D_{xy}. 
\end{align*}
Let $\tu{x,y}\in\coeq((\exists z)\varphi$. From $x\neq y$ follows that $x\neq z$, $y\neq z$, and $\tu{x,y}\in\coeq(\varphi)$. We have
\begin{align*}
\overline{||(\exists z)\varphi||}=\overline{\C_z(||\varphi||)}\subseteq\overline{||\varphi||}\subseteq\D_{xy}. 
\end{align*}
We showed that the assertion holds for existentially quantified formulas.
\end{proof}

Equality and coequality of variables in formulas become trivial for certain sets of formulas. \emph{Positive formulas}  are inductively defined by the following rules.
\begin{enumerate}
\item each atomic formula is positive;
\item if $\varphi$ is not positive, then $\neg\varphi$ is positive;
\item if $\varphi$ is positive, then $(\exists x)\varphi$ is positive;
\item if $\varphi$ or $\psi$ is positive, then $\varphi\wedge\psi$ is positive.
\end{enumerate}
Let us recall that formulas of the form $\varphi\wedge\psi$ are, in the previous definition, treated as shortcut for $\neg(\neg\varphi\wedge\neg\psi)$. 
Formulas that are not positive are called \emph{negative}.

The following lemma follows directly from the definitions of equality and coequality of variables in formulas.
\begin{lem}\label{lem:pos_neg_coeq_eq}
If $\varphi$ is a positive or negative formula, then $\coeq(\varphi)$ or $\eq(\varphi)$, respectively, is the identity on $\X$. 
\end{lem}

We can infer from the structure of a given formula the set of variables for which the value of the formula has elements only in the active domain.

For an equivalence $E$ on $\X$ and a subset $X$ of $\X$, we denote by $\cl_E(X)$ the set of variables given by
\begin{align}
\cl_E(X)=\{y\in \X\st x\in X \textrm{ and }\tu{x,y}\in E\}. 
\end{align}

For a formula $\varphi$, we define sets of variables $\genz(\varphi)$ and $\cogenz(\varphi)$ by the following rules:
\begin{enumerate}
\item $\genz(1)=\cogenz(1)=\emptyset,$
\item $\genz(r(x_1,\ldots,x_n))=\{x_1,\ldots,x_n\},$
\item $\cogenz(r(x_1,\ldots,x_n))=\emptyset,$
\item $\genz(x\approx y)=\emptyset$,
\item $\cogenz(x\approx y)=\emptyset$,
\item $\genz(\neg\varphi)=\cogenz(\varphi),$
\item $\cogenz(\neg\varphi)=\genz(\varphi),$
\item $\genz(\varphi\wedge\psi)=\cl_E(\genz(\varphi)\cup\genz(\psi))$, where $E=\eq(\varphi\wedge\psi)$,
\item $\cogenz(\varphi\wedge\psi)=\cogenz(\varphi)\cap\cogenz(\psi),$
\item $\genz((\exists x)\varphi)=\genz(\varphi)\setminus\{x\}$,
\item $\cogenz((\exists x)\varphi)=\cogenz(\varphi)\setminus\{x\}$.
\end{enumerate}

Prior to proving the following theorem, describing the role of the mappings $\genz$ and $\cogenz$, 
we establish some auxiliary lemmas.

\begin{lem}\label{lem:cap_project}
Let $X_1,X_2\subseteq\X$ and $V_1,V_2\subseteq\V$. Then it holds
\begin{align*}
\epsilon(\pi_{X_1\cup X_2}(V_1\cap V_2))\subseteq\epsilon(\pi_{X_1}(V_1))\cap\epsilon(\pi_{X_2}(V_2)).
\end{align*}
\end{lem}
\begin{proof}
Let $v\in\epsilon(\pi_{X_1\cup X_2}(V_1\cap V_2))$. Then there exists $v'\in V_1\cap V_2$ such that $v|_{X_1\cup X_2}=v'|_{X_1\cup X_2}$. Clearly, $v'|_{X_1}\in\pi_{X_1}(V_1)$, $v'|_{X_2}\in\pi_{X_2}(V_2)$, $v|_{X_1}=v'|_{X_1}$, and $v|_{X_2}=v'|_{X_2}$. Thus, $v\in\epsilon(\pi_{X_1}(V_1))\cap\epsilon(\pi_{X_2}(V_2))$. 
\end{proof}

\begin{lem}\label{lem:cap_dom}
Let $N\subseteq M$, $X_1,X_2\subseteq\X$. Then it holds
\begin{align*}
\epsilon(N^{X_1\cup X_2})=\epsilon(N^{X_1})\cap\epsilon(N^{X_2}). 
\end{align*}
\end{lem}
\begin{proof}
Let $v\in\V$. Then $v\in\epsilon(N^{X_1\cup X_2})$ if{}f there exists $t\in N^{X_1\cup X_2}$ such that $v|_{X_1\cup X_2}=t$ if{}f there exists $t_1\in N^{X_1}$ and $t_2\in N^{X_2}$ such that $v|_{X_1}=t_1$ and $v|_{X_2}=t_2$ if{}f $v\in\epsilon(N^{X_1})\cap\epsilon(N^{X_2})$. Note that for $t_1\in N^{X_1}$ and $t_2\in N^{X_2}$, we have $t_1\cup(t_2|_{X_2\setminus X_1})\in N^{X_1\cup X_2}$.
\end{proof}

\begin{lem}\label{lem:eq_dom}
Let $N\subseteq M$, $X\subseteq\X$, and $\varphi$ by a formula such that $\pi_X(||\varphi||)\subseteq N^X$. Then it holds
\begin{align*}
\pi_{\cl_{\eq(\varphi)}(X)}(||\varphi||)\subseteq N^{\cl_{\eq(\varphi)}(X)}. 
\end{align*}
\end{lem}
\begin{proof}
Let $t\in\pi_{\cl_{\eq(\varphi)}(X)}(||\varphi||)$ and $x\in\cl_{\eq(\varphi)}(X)$. Then there exists $y\in X$ such that $\tu{x,y}\in\eq(\varphi)$. Thus, $||\varphi||\subseteq\D_{xy}$.  Now, $t|_X\in\pi_X(||\varphi||)$, and therefore $t|_X\in N^X$. We have $t(y)\in N$, and since $t(y)=t(x)$, it follows $t(x)\in N$ as well. We have thus shown that $t\in N^{\cl_{\eq(\varphi)}(X)}$.
\end{proof}

The following observations are immediate:
\begin{align}
&\pi_Y(\cl_X(V))=\pi_Y(V),&\textrm{ provided that } X\cap Y=\emptyset, \label{eq:pi_C_disjoint}\\
&\pi_Y(\overline{\cl_X(V)})\subseteq\pi_Y(\overline{V}).& \label{eq:pi_overline_C}
\end{align}

We can now state a theorem that describes the meaning of the mappings $\genz$ and $\cogenz$.
\begin{thm}
For every formula $\varphi$, it holds
\begin{align*}
\pi_{\genz(\varphi)}(||\varphi||)&\subseteq \A(M)^{\genz(\varphi)},\\
\pi_{\cogenz(\varphi)}(\overline{||\varphi||})&\subseteq \A(M)^{\cogenz(\varphi)}. 
\end{align*}
\end{thm}
\begin{proof}
We prove the assertion by structural induction on formulas. Clearly, it holds for any atomic formula. Suppose that the assertion holds for formulas $\varphi$ and $\psi$. We have
\begin{align*}
\pi_{\genz(\neg\varphi)}(||\neg\varphi||)=\pi_{\cogenz(\varphi)}(\overline{||\varphi||})\subseteq\A(M)^{\cogenz(\varphi)}=\A(M)^{\genz(\neg\varphi)}
\end{align*}
and 
\begin{align*}
\pi_{\cogenz(\neg\varphi)}(\overline{||\neg\varphi||})&=\pi_{\genz(\varphi)}(\overline{\overline{||\varphi||}})=\pi_{\genz(\varphi)}(||\varphi|||)\\
&\subseteq\A(M)^{\genz(\varphi)}=\A(M)^{\cogenz(\neg\varphi)}.
\end{align*}
Therefore, the assertion holds for negated formulas.
From the assumptions $\pi_{\genz(\varphi)}(||\varphi||)\subseteq \A(M)^{\genz(\varphi)}$ and $\pi_{\genz(\psi)}(||\psi||)\subseteq \A(M)^{\genz(\psi)}$, it follows that 
\begin{align*}
\epsilon(\pi_{\genz(\varphi)}(||\varphi||))\cap\epsilon(\pi_{\genz(\psi)}(||\psi||))\subseteq\epsilon(\A(M)^{\genz(\varphi)})\cap\epsilon(\A(M)^{\genz(\psi)}).
\end{align*}
From Lemmas \ref{lem:cap_project} and \ref{lem:cap_dom}, it follows that
\begin{align*}
\pi_{\genz(\varphi)\cup\genz(\psi)}(||\varphi\wedge\psi||)\subseteq\A(M)^{\genz(\varphi)\cup\genz(\psi)}. 
\end{align*}
From Lemma \ref{lem:eq_dom}, we obtain
\begin{align*}
\pi_{\genz(\varphi\wedge\psi)}(||\varphi\wedge\psi||)\subseteq\A(M)^{\genz(\varphi\wedge\psi)}. 
\end{align*}
From the assumptions 
$$\pi_{\cogenz(\varphi)}(\overline{||\varphi||})\subseteq \A(M)^{\cogenz(\varphi)}$$ 
and 
$$\pi_{\cogenz(\psi)}(\overline{||\psi||})\subseteq \A(M)^{\cogenz(\psi)}$$
together with the equality $\cogenz(\varphi\wedge\psi)=\cogenz(\varphi)\cap\cogenz(\psi)$, we obtain 
\begin{align*}
\pi_{\cogenz(\varphi\wedge\psi)}(\overline{||\varphi||})&\subseteq \A(M)^{\cogenz(\varphi\wedge\psi)},\\
\pi_{\cogenz(\varphi\wedge\psi)}(\overline{||\psi||})&\subseteq \A(M)^{\cogenz(\varphi\wedge\psi)}. 
\end{align*}
It follows that
\begin{align*}
\pi_{\cogenz(\varphi\wedge\psi)}(\overline{||\varphi\wedge\psi||})&= \pi_{\cogenz(\varphi\wedge\psi)}(\overline{||\varphi||}\cup\overline{||\psi||})\\
&=\pi_{\cogenz(\varphi\wedge\psi)}(\overline{||\varphi||})\cup\pi_{\cogenz(\varphi\wedge\psi)}(\overline{||\psi||})\\
&\subseteq\A(M)^{\cogenz(\varphi\wedge\psi)}.
\end{align*}
We have shown that the assertion also holds for conjoined formulas.

From the assumption $\pi_{\genz(\varphi)}(||\varphi||)\subseteq \A(M)^{\genz(\varphi)}$, we obtain
\begin{align*}
\pi_{\genz(\varphi)\setminus\{x\}}(||\varphi||)\subseteq \A(M)^{\genz(\varphi)\setminus\{x\}}. 
\end{align*}
Using the equality \eqref{eq:pi_C_disjoint}, it follows that
\begin{align*}
\pi_{\genz((\exists x)\varphi)}(||(\exists x)\varphi||)&=\pi_{\genz(\varphi)\setminus\{x\}}(\C_x(||\varphi||))=\pi_{\genz(\varphi)\setminus\{x\}}(||\varphi||)\\
&\subseteq\A(M)^{\genz(\varphi)\setminus\{x\}}=\A(M)^{\genz((\exists x)\varphi)}.
\end{align*}

From the assumption $\pi_{\cogenz(\varphi)}(\overline{||\varphi||})\subseteq \A(M)^{\cogenz(\varphi)}$, we obtain
\begin{align*}
\pi_{\cogenz(\varphi)\setminus\{x\}}(\overline{||\varphi||})\subseteq \A(M)^{\cogenz(\varphi)\setminus\{x\}}. 
\end{align*}
Using the equality \eqref{eq:pi_overline_C}, it follows that
\begin{align*}
\pi_{\cogenz((\exists x)\varphi)}(\overline{||(\exists x)\varphi||})&=\pi_{\cogenz(\varphi)\setminus\{x\}}(\overline{\C_x(||\varphi||)})\subseteq\pi_{\cogenz(\varphi)\setminus\{x\}}(||\varphi||)\\
&\subseteq\A(M)^{\cogenz(\varphi)\setminus\{x\}}=\A(M)^{\cogenz((\exists x)\varphi)}.
\end{align*}
We have shown that the assertion also holds for existentially quantified formulas.
\end{proof}

Following Gelder and Topor \cite{GelderTopoer:safety}, we define the allowed formulas accordingly.
A formula $\varphi$ is called \emph{allowed} if $\FV(\varphi)=\genz(\varphi)$ and every subformula of $\varphi$ of the form $(\exists x)\psi$ satisfies $x\in\genz(\psi)$.

\section{Normalization}\label{sec:normalization}
This section addresses the construction of a normalized formula semantically equivalent to a given allowed formula. We start by introducing a derived logical connective. We denote by $\varphi\vee^*\psi$  the formula  $(\exists X\setminus Y)\varphi\vee(\exists Y\setminus X)\psi,$ where $X=\FV(\varphi)$ and $Y=\FV(\psi)$.

We describe, for a given formula, the normalized formula which restricts elements of its value. For a formula $\varphi$ we define formulas $\gen(\varphi)$ and $\cogen(\varphi)$, called the \emph{generator} and \emph{cogenerator} of $\varphi$, respectively, by the following rules.
\begin{enumerate}
\item $\gen(1)=\cogen(1)=1$;
\item $\gen(r(x_1,\ldots,x_n))=r(x_1,\ldots,x_n)$;
\item $\cogen(r(x_1,\ldots,x_n))=1$;
\item $\gen(x\approx y)=1$;
\item $\cogen(x\approx y)=1$;
\item $\gen(\neg\varphi)=\cogen(\varphi)$;
\item $\cogen(\neg\varphi)=\gen(\varphi)$;
\item $\gen(\varphi\wedge\psi)=\gen(\varphi)\wedge\gen(\psi)\wedge(\varphi\approx_\wedge\psi)$;
\item $\cogen(\varphi\wedge\psi)=\cogen(\varphi)\vee^*\cogen(\psi)$;
\item $\gen((\exists x)\varphi)=(\exists x)\gen(\varphi)$;
\item $\cogen((\exists x)\varphi)=(\exists x)\cogen(\varphi)$.
\end{enumerate}
It remains to define the formula $\varphi\approx_\wedge\psi$ occurring in the eighth point. A \emph{minimal representation} of an equivalence $E$ on $\X$ is a minimal binary relation $R$ on $\X$ (with respect to the set inclusion) such that $\Eq(R)=E$. We denote by $(\varphi_1\approx_\wedge\varphi_2)$ the formula $(x_1\approx y_1)\wedge\cdots\wedge (x_n\approx y_n)$, where $\{\tu{x_i,y_i}\in\X\times\X\st 1\leq i \leq n\}$ is a minimal representation of $\Eq(\eq(\varphi_1\wedge\varphi_2)\setminus\eq(\gen(\varphi_1)\wedge\gen(\varphi_2)))$. If the minimal representation is the empty set, we set $(\varphi_1\approx_\wedge\varphi_2)$ equal to $1$. For each equivalence on $\X$, we fix one of its minimal representations. For example, if $\varphi_1=r(x_1,x_2)$ and $\varphi_2=x_1\approx x_2$, then $\varphi_1\approx_\wedge\varphi_2=x_1\approx x_2$, since $\eq(\varphi_1\wedge\varphi_2)=\Eq(\{\tu{x_1,x_2}\})$ and $\eq(\gen(\varphi_1)\wedge\gen(\varphi_2))=\eq(r(x_1,x_2)\wedge1)=\Eq(\emptyset)$.

Mappings $\gen$ and $\cogen$ are generalisations of $\genz$ and $\cogenz$, respectively, in the following sense:
$\FV(\gen(\varphi))=\genz(\varphi)$ and $\FV(\cogen(\varphi))=\cogenz(\varphi)$ holds for every formula $\varphi$.

For example, consider relational division, which can be expressed by the formula $\varphi=(\exists y)s(x,y)\wedge(\forall y)(r(y)\to s(x,y))$. Without abbreviations, the formula $\varphi$ is $(\exists y)s(x,y)\wedge\neg(\exists y)\neg(\neg r(y)\vee s(x,y))$. We have
\begin{align*}
&\gen((\exists y)s(x,y)\wedge\neg(\exists y)\neg(\neg r(y)\vee s(x,y)))\\
=&(\exists y)s(x,y)\wedge\gen(\neg r(y)\vee s(x,y))\wedge 1\\
=&(\exists y)s(x,y)\wedge\gen(\neg(\neg\neg r(y)\wedge\neg s(x,y)))\wedge 1\\
=&(\exists y)s(x,y)\wedge(\cogen(\neg\neg r(y))\vee^*\cogen(\neg s(x,y)))\wedge 1\\
=&(\exists y)s(x,y)\wedge(1\vee^*s(x,y))\wedge 1\\
=&(\exists y)s(x,y)\wedge(1\vee(\exists x)(\exists y)s(x,y))\wedge 1.
\end{align*}
Although it is not necessary, we can simplify the result by the following rules, which we call \emph{tautology rules}:
\begin{enumerate}
\item $(\exists x)1$ simplifies to $1$,
\item $1\wedge \varphi$ and $\varphi\wedge 1$ simplify to $\varphi$,
\item $1\vee\psi$ and $\varphi\vee 1$ simplify to $1$.  
\end{enumerate}
After simplification, the generator of $\varphi$ is $(\exists y)s(x,y)$.

The following lemma shows that the cogenerators of positive formulas and the generators of negative formulas are trivial.
\begin{lem}\label{lem:pos_neg_cogen_gen}
If $\varphi$ is positive, then $||\cogen(\varphi)|| = \V$; if it is negative, then $||\gen(\varphi)|| = \V$.
\end{lem}
\begin{proof}
The assertion holds trivially for atomic formulas, as they are positive. Suppose that $\varphi$ is positive. Then $\neg\varphi$ is negative, and we have $||\gen(\neg\varphi)||=||\cogen(\varphi)||=\V$. Similarly, it can be shown that the assertion holds also for $\neg\varphi$, provided that $\varphi$ is negative. 

Suppose that $\varphi$ is positive. Then $\varphi\wedge\psi$ is positive as well. We have $||\cogen(\varphi\wedge\psi)||=||\cogen(\varphi)\vee^*\cogen(\psi)||=||1\vee^*\cogen(\psi)||=\V$. The same conclusion holds under the assumption that $\psi$ is positive.

Suppose that both $\varphi$ and $\psi$ are negative. Then $\varphi\wedge\psi$ is also negative and, by Lemma \ref{lem:pos_neg_coeq_eq}, $\eq(\varphi\wedge\psi)$ is the identity on $\X$. Therefore, we have $||\gen(\varphi\wedge\psi)||=||\gen(\varphi)||\cap||\gen(\psi)||=\V\cap\V=\V$.

Suppose that $\varphi$ is positive. Then $(\exists x)\varphi$ is also positive. We have $||\cogen((\exists x)\varphi)||=\C_x(||\cogen(\varphi)||)=\C_x(\V)=\V$. Similarly, if  $\varphi$ is negative, then $||\gen((\exists x)\varphi)||=\V$. 
\end{proof}

At this point, we can precisely describe how generators and cogenerators restrict values of formulas.
\begin{thm}\label{thm:gen}
It holds
\begin{align*}
||\varphi||&\subseteq||\gen(\varphi)||,\\
\overline{||\varphi||}&\subseteq||\cogen(\varphi)||. 
\end{align*}
 \end{thm} 
\begin{proof}
We prove the assertion by structural induction on formulas. Clearly, it holds for any atomic formula. 

Suppose that the assertion holds for formulas $\varphi$ and $\psi$. We have 
\begin{align*}
||\neg\varphi||=\overline{||\varphi||}\subseteq||\cogen(\varphi)||=||\gen(\neg\varphi)||
\end{align*}
and
\begin{align*}
\overline{||\neg\varphi||}=\overline{\overline{||\varphi||}}=||\varphi||\subseteq||\gen(\varphi)||=||\cogen(\neg\varphi)||. 
\end{align*}
We showed that the assertion holds also for negated formulas. We have
\begin{align*}
||(\exists x)\varphi||=\C_x(||\varphi||)\subseteq\C_x(||\gen(\varphi)||)=||(\exists x)\gen(\varphi)||=||\gen((\exists x)\varphi)||
\end{align*}
and
\begin{align*}
\overline{||(\exists x)\varphi||}&=\overline{\C_x(||\varphi||)}\subseteq\overline{||\varphi||}\subseteq\C_x(\overline{||\varphi||})\subseteq\C_x(||\cogen(\varphi)||)\\
&=||(\exists x)\cogen(\varphi)||=||\cogen((\exists x)\varphi)||.
\end{align*}
We have shown that the assertion also holds for existentially quantified formulas.

By Lemma \ref{lem:eq_coeq}, we have
\begin{align*}
||\varphi\wedge\psi||&=||\varphi||\cap||\psi||=||\varphi||\cap||\psi||\cap\D_{x_1y_1}\cap\cdots\cap\D_{x_ny_n}\\
&\subseteq||\gen(\varphi)||\cap||\gen(\psi)||\cap\D_{x_1y_1}\cap\cdots\cap\D_{x_ny_n}\\
&=||\gen(\varphi)\wedge\gen(\psi)\wedge(x_1\approx y_1)\wedge\cdots\wedge (x_n\approx y_n)||\\
&=||\gen(\varphi\wedge\psi)||,
\end{align*}
 where $\{\tu{x_1,y_1},\ldots,\tu{x_n,y_n}\}$ is a minimal representation of $\Eq(\eq(\varphi_1\wedge\varphi_2)\setminus\eq(\gen(\varphi_1)\wedge\gen(\varphi_2)))$.

We have
\begin{align*}
\overline{||\varphi\wedge\psi||}&= \overline{||\varphi||\cap||\psi||}=\overline{||\varphi||}\cup\overline{||\psi||}\subseteq||\cogen(\varphi)||\cup||\cogen(\psi)||\\
&\subseteq\C_{X\setminus Y}(||\cogen(\varphi)||)\cup\C_{Y\setminus X}(||\cogen(\psi)||)\\
&=||(\exists X\setminus Y)\cogen(\varphi)\vee(\exists Y\setminus X)\cogen(\psi)||\\
&=||\cogen(\varphi\wedge\psi)||,
\end{align*}
where $X=\FV(\cogen(\varphi))$ and $Y=\FV(\cogen(\psi))$.

We have shown that the assertion also holds for conjoined formulas.
\end{proof}

We next consider the normalization of formulas. To handle double negation correctly during normalization, we introduce the following notion of the complement of a formula. For a formula $\varphi$, we define the \emph{complement} $\overline \varphi$ of $\varphi$ by
\[
\overline\varphi =
\begin{cases} 
\neg\varphi\quad&\textrm{ if }\varphi\textrm{ is not negated },\\
\varphi_1\quad&\textrm{ if }\varphi\textrm{ is a formula of the form }\neg\varphi_1.
\end{cases}
\] 

Clearly, $||\overline\varphi||=\overline{||\varphi||}$, $\eq(\overline\varphi)=\coeq(\varphi)$, and $\coeq(\overline\varphi)=\eq(\varphi)$.

During normalization, we will need to align the free variables of a given formula. Let $\varphi$ and $\psi$ be formulas satisfying $\FV(\varphi)\subseteq\FV(\psi)$. Then, by \emph{alignment} of a formula $\varphi$ via $\psi$, we mean the formula $\al(\varphi,\psi)$ given by:
\begin{enumerate}
\item $\al(\varphi,\psi)=\varphi$, if $\FV(\varphi)=\FV(\psi)$;
\item $\al(\varphi,\psi)=\varphi\wedge(\exists \FV(\varphi))\psi$, otherwise.  
\end{enumerate}
The following containment is an immediate consequence of the definition:
\begin{align}
||\al(\varphi,\psi)||\subseteq||\varphi||.\label{eq:align_subseteq}
\end{align}
Clearly, $\eq(\al(\varphi,\psi))\supseteq\eq(\varphi)$.

\begin{lem}\label{lem:align}
If $\FV(\varphi)\subseteq\FV(\psi)$, then  $||\al(\varphi,\psi)||\cap||\psi||=||\varphi||\cap||\psi||$.
\end{lem}
\begin{proof}
If $\FV(\varphi)=\FV(\psi)$, then the assertion holds trivially.
Suppose that $\FV(\varphi)\subset \FV(\psi)$. Then we have
\begin{align*}
||\al(\varphi,\psi)||\cap||\psi||&=||\varphi\wedge(\exists \FV(\varphi))\psi||\cap||\psi||=||\varphi||\cap\C_{\FV(\varphi)}(||\psi||)\cap||\psi||\\
&=||\varphi||\cap||\psi||.  \qedhere
\end{align*} 
 
\end{proof}

We are ready to describe how a formula is normalized.
We assign to a formula $\varphi$ and a normalized formula $\psi$ such that $\FV(\varphi)=\FV(\psi)$ the formula $\norm(\varphi,\psi)$ called the \emph{normalization} of $\varphi$ using $\psi$ by the following rules:
\begin{enumerate}
\item if $\varphi=r(x_1,\ldots,x_n)$, then $\norm(\varphi,\psi)=r(y_1/x_1,\ldots,y_n/x_n)$, where $r(y_1,\ldots,y_n)\in\F$; 
\item if $\varphi=x_1\approx x_2$, then $\norm(\varphi,\psi)=\psi\wedge(x_1\approx x_2)$;
\item if $\varphi=\neg\varphi_1$, then $\norm(\varphi,\psi)=\overline{\norm(\varphi_1,\psi)}$;
\item if $\varphi=\varphi_1\wedge\varphi_2$, then denote $\alpha_1=\norm(\varphi_1,\psi_1)$, $\alpha_2=\norm(\varphi_2,\psi_2)$, where $\psi_1=(\exists\FV(\psi)\setminus\FV(\varphi_1))\psi$ and $\psi_2=(\exists\FV(\psi)\setminus\FV(\varphi_2))\psi$, and 
\begin{enumerate}
\item  if both $\alpha_1$ and $\alpha_2$ are not negated, then $\norm(\varphi,\psi)=\alpha_1\wedge\alpha_2$;
\item  if  $\alpha_1$ is not negated and $\alpha_2$ is negated, then 
$$\norm(\varphi,\psi)=\al(\alpha_1,\psi)\wedge\neg\al(\overline{\alpha_2},\psi);$$ 
\item if  $\alpha_2$ is not negated and $\alpha_1$ is negated, then 
$$\norm(\varphi,\psi)=\al(\alpha_2,\psi)\wedge\neg\al(\overline{\alpha_1},\psi);$$

\item  if both $\alpha_1$  and $\alpha_2$ are negated, then 
$$\norm(\varphi,\psi)=\neg(\al(\overline{\alpha_1},\psi)\vee\al(\overline{\alpha_2},\psi));$$
\end{enumerate}
\item if $\varphi=(\exists x)\varphi_1$, then
$$\norm(\varphi,\psi)=(\exists x)\norm(\varphi_1,(\exists X)\psi_1\wedge\psi),$$
where $\psi_1=\gen(\varphi_1)$ and $X=\FV(\psi_1)\setminus\{x\}$.
\end{enumerate}
By inspection of the rules just described, we see that the formula $\norm(\varphi,\psi)$ is normalized. By the \emph{normalisation} of a formula $\varphi$, we mean the formula $\norm(\varphi,\gen(\varphi))$. For example, we compute the normalisation of the relational division formula discussed in Section~\ref{sec:normalization}:
\begin{align}
\varphi_1=(\exists y)s(x,y)\wedge\neg(\exists y)(r(y)\wedge\neg s(x,y)).
\end{align}
In order to improve the readability of the example, we use simplification rules for the tautology $1$ and simplifications made during the normalization of atomic formulas. We already know that the generator $\gen(\varphi_1)$ can be simplified to $\psi_1=(\exists y)s(x,y)$. Normalization proceeds as follows.
\begin{enumerate}
\item Since $\varphi_1$ is equal $\varphi_2\wedge\varphi_3$, where $\varphi_2=(\exists y)s(x,y)$ and $\varphi_3=\neg(\exists y)(r(y)\wedge\neg s(x,y))$, we need to compute $\norm(\varphi_2,\psi_1)$ and $\norm(\varphi_3,\psi_1)$. 
\item We have $\norm(\varphi_2,\psi_1)=(\exists y)\norm(s(x,y),(\exists x)s(x,y)\wedge\psi_1)=s(x/x,y/y)$, which simplifies to $s(x,y)$.
\item We have $\norm(\varphi_3,\psi_1)=\overline{\norm(\varphi_4,\psi_1)}$, where $\varphi_4$ is $(\exists y)(r(y)\wedge\neg s(x,y))$.
\item In order to compute $\norm(\varphi_4,\psi_1)$, we need to compute $\gen(r(y)\wedge\neg s(x,y))=r(y)\wedge1$, which simplifies to $r(y)$. Now,
\begin{align*}
\norm(\varphi_4,\psi_1)=(\exists y)(\norm(r(y)\wedge\neg s(x,y),r(y)\wedge\psi_1)).
\end{align*}
We need to compute $\norm(r(y),(\exists x)(r(y)\wedge\psi_1))$ and $\norm(\neg s(x,y),r(y)\wedge\psi_1)$.
\item We have $\norm(r(y),(\exists x)(r(y)\wedge\psi_1))=r(y/y)$, which simplifies to $r(y)$ \\and $\norm(\neg s(x,y),r(y)\wedge\psi_1)=\overline{\norm(s(x,y),r(y)\wedge\psi_1)}=\neg{s(x/x,y/y)}$, which simplifies to $\neg s(x,y)$.
\item We can continue with step four and obtain
\begin{align*}
\norm(\varphi_4,\psi_1)&=(\exists y)\norm(r(y)\wedge\neg s(x,y),r(y)\wedge\psi_1)\\
&=(\exists y)(\al(r(y),r(y)\wedge\psi_1)\wedge\neg\al(s(x,y),r(y)\wedge\psi_1))\\
&=(\exists y)((r(y)\wedge(\exists y)(r(y)\wedge\psi_1))\wedge\neg s(x,y)). 
\end{align*}
Denote the last formula by $\varphi_5$. Therefore,
\begin{align*}
\norm(\varphi_3,\psi_1)=\overline{\norm(\varphi_4,\psi_1)}=\neg\varphi_5.
\end{align*}
\item We can continue with the first step:
\begin{align*}
\norm(\varphi_1,\psi_1)&=\norm(\varphi_2\wedge\varphi_3,\psi_1)\\
&=\al((\exists y)s(x,y),\psi_1)\wedge\neg\al(\varphi_5,\psi_1)\\
&=(\exists y)s(x,y)\wedge\neg\varphi_5\\
&=(\exists y)s(x,y)\wedge\neg(\exists y)((r(y)\wedge(\exists y)(r(y)\wedge(\exists y)s(x,y)))\wedge\neg s(x,y))
\end{align*}
\end{enumerate}
Using third and fourth axiom of cylindric algebras we can simplify $\norm(\varphi_1,\psi_1)$ and obtain:
$$\varphi_6=(\exists y)s(x,y)\wedge\neg(\exists y)((r(y)\wedge(\exists y)s(x,y))\wedge\neg s(x,y))$$
Since $\varphi_6$ is normalized, we can obtain semantically equivalent relational expression:
\begin{align*}
\expr(\varphi_6)=\pi_{\{x\}}(s)-\pi_{\{x\}}((r\bowtie\pi_{\{x\}}(s))-s).
\end{align*}
Which is the standard definition of relational division.

We observe that $\norm(\varphi,\psi)$ is positive if it is not negated, and negative otherwise. Another observation is that $\varphi$ is positive if and only if $\norm(\varphi, \psi)$ is positive.

The following lemma shows that normalization preserves equality and coequality of variables in formulas.
\begin{lem}\label{lem:eq_coeq_norm}
It holds $\eq(\varphi)\subseteq\eq(\norm(\varphi,\psi))$ and $\coeq(\varphi)\subseteq\coeq(\norm(\varphi,\psi))$.  
\end{lem}
\begin{proof}
We prove the assertion by structural induction on the formula $\varphi$.
\begin{enumerate}
\item Cases where $\varphi$ is of the form $r(x_1,\ldots,x_n)$ or $1$ are trivial, since in those cases $\eq(\varphi)=\coeq(\varphi)=\id_\X$.
\item Consider the case where $\varphi$ is of the form $x_1\approx x_2$. Then we have $\eq(\norm(x_1\approx x_2,\psi))=\eq(\psi\wedge(x_1\approx x_2))=\Eq(\eq(\psi)\cup\eq(x_1\approx x_2))\supseteq\eq(x_1\approx x_2)$. Clearly, $\coeq(x_1\approx x_2)=\id_\X\subseteq\coeq(\norm(\varphi,\psi))$.
\item Consider the case where $\varphi$ is of the form $\neg\varphi_1$. Then we have $\eq(\neg\varphi_1)=\coeq(\varphi_1)\subseteq\coeq(\norm(\varphi_1,\psi))=\eq(\overline{\norm(\varphi_1,\psi)})=\eq(\norm(\neg\varphi_1,\psi))$. Similarly can be shown that $\coeq(\neg\varphi_1)\subseteq\coeq(\norm(\neg\varphi_1,\psi))$.
\item Consider the case where $\varphi$ is of the form $\varphi_1\wedge\varphi_2$. Let $\alpha_1=\norm(\varphi_1,\psi_1)$ and $\alpha_2=\norm(\varphi_2,\psi_2)$, where $\psi_1=(\exists\FV(\psi)\setminus\FV(\varphi_1))\psi$ and $\psi_2=(\exists\FV(\psi)\setminus\FV(\varphi_2))\psi$. We distinguish four subcases. By the Lemma \ref{lem:pos_neg_coeq_eq}, it suffices to prove only one part of the claim, since the other is trivially true.
\begin{enumerate}
\item Consider the subcase in which both $\alpha_1$ and $\alpha_2$ are not negated. We have $\eq(\norm(\varphi_1\wedge\varphi_2))=\eq(\alpha_1\wedge\alpha_2)=\Eq(\eq(\alpha_1)\cup\eq(\alpha_2))\supseteq\Eq(\eq(\varphi_1)\cup\eq(\varphi_2))=\eq(\varphi_1\wedge\varphi_2)$.
\item Consider the subcase in which $\alpha_1$ is not negated and $\alpha_2$ is negated. We have 
\begin{align}
&\eq(\norm(\varphi_1\wedge\varphi_2,\psi))=\eq(\al(\alpha_1,\psi)\wedge\neg\al(\overline{\alpha_2},\psi))\notag\\
&=\Eq(\eq(\al(\alpha_1,\psi))\cup\eq(\neg\al(\overline{\alpha_2},\psi))) \notag\\
&=\eq(\al(\alpha_1,\psi))\label{eq:lem_eq_coeq_norm1}\\
&\supseteq\eq(\alpha_1)\supseteq\eq(\varphi_1)\notag\\
&=\Eq(\eq(\varphi_1)\cup\eq(\varphi_2))\label{eq:lem_eq_coeq_norm2}\\
&=\eq(\varphi_1\wedge\varphi_2).\notag
\end{align}
Since $\alpha_2$ is negative, $\neg\al(\overline{\alpha_2},\psi)$ is also negative. Thus, $\eq(\neg\al(\overline{\alpha_2},\psi))$ is the identity on $\X$, and the equality \eqref{eq:lem_eq_coeq_norm1} holds. Similarly, the equality \eqref{eq:lem_eq_coeq_norm2} holds, since $\varphi_2$ is negative.
\item The subcase where $\alpha_2$ is not negated and $\alpha_1$ is negated proceeds analogously to the previous one. 
\item Consider the subcase in which both $\alpha_1$ and $\alpha_2$ are negated. We have $\coeq(\norm(\varphi_1\wedge\varphi_2,\psi))=\coeq(\neg(\al(\overline{\alpha_1},\psi)\vee\al(\overline{\alpha_2},\psi)))=\coeq(\neg\al(\overline{\alpha_1},\psi)\wedge\neg\al(\overline{\alpha_2},\psi))\\=\eq(\al(\overline{\alpha_1},\psi))\cap\eq(\al(\overline{\alpha_2},\psi))\supseteq\eq(\overline{\alpha_1})\cap\eq(\overline{\alpha_2})=\coeq(\alpha_1)\cap\coeq(\alpha_2)\supseteq\coeq(\varphi_1)\cap\coeq(\varphi_2)=\coeq(\varphi_1\wedge\varphi_2)$.
\end{enumerate}
\item Consider the case where $\varphi$ is of the form $(\exists x)\varphi_1$. Let $\psi_1=\gen(\varphi_1)$ and $X=\FV(\psi_1)\setminus\{x\}$. We have $\eq(\norm((\exists x)\varphi_1,\psi))=\eq((\exists x)\norm(\varphi_1,(\exists X)\psi_1\wedge\psi))=\Eq(\eq(\norm(\varphi_1,(\exists X)\psi_1\wedge\psi))|_{(X\setminus\{x\})^2})\supseteq\Eq(\eq(\varphi_1)|_{(X\setminus\{x\})^2})=\eq((\exists x)\varphi_1)$. Similarly can be shown that $\coeq(\norm((\exists x)\varphi_1,\psi))\supseteq\eq((\exists x)\varphi_1)$.  \qedhere
\end{enumerate}
\end{proof}
Normalizing a conjunction preserves the equality of variables expressed by the connective $\approx_\wedge$.
\begin{lem}\label{lem:norm_eq_wedge}
It holds $||\norm(\varphi_1\wedge\varphi_2,\psi)||\subseteq||\varphi_1\approx_\wedge\varphi_2||$.   \qedhere
\end{lem}
\begin{proof}
Suppose that $\varphi_1\approx_\wedge\varphi_2=(x_1\approx x_2)\wedge\cdots\wedge(x_n\approx x_n)$. By the definition of $\varphi_1\approx_\wedge\varphi_2$, we have $\tu{x_i,y_i}\in\eq(\varphi_1\wedge\varphi_2)$ for all $1\leq i \leq n$. It follows from Lemma \ref{lem:eq_coeq_norm} that $\tu{x_i,y_i}\in\eq(\norm(\varphi_1\wedge\varphi_2,\psi))$ for all $1\leq i \leq n$. 
We have, by Lemma \ref{lem:eq_coeq}, 
\begin{align*}
||\norm(\varphi_1\wedge\varphi_2,\psi)||\subseteq\D_{x_1y_1}\cap\cdots\cap\D_{x_ny_n}=||\varphi_1\approx_\wedge\varphi_2||.  
\end{align*}
\qedhere
\end{proof}
The normalization of a formula always remains within the space described by its generator.
\begin{lem}\label{lem:norm1}
It holds that 
$$||\norm(\varphi,\psi)||\subseteq||\gen(\varphi)||$$ 
and 
$$\overline{||\norm(\varphi,\psi)||}\subseteq||\cogen(\varphi)||.$$ 
\end{lem}
\begin{proof}
We prove the assertion by structural induction on the formula $\varphi$.
\begin{enumerate}
\item Consider the case in which $\varphi$ is of the form $r(x_1,\ldots,x_n)$.  Then we have\\ $||\norm(r(x_1,\ldots,x_n),\psi)||=||r(y_1/x_1,\ldots,y_n/x_n)||$, where $r(y_1,\ldots,y_n)\in\F$, by Lemma \ref{lem:rename}, $||r(y_1/x_1,\ldots,y_n/x_n)||=||r(x_1,\ldots,x_n)||$. Now, by Lemma \ref{thm:gen}, $||r(x_1,\ldots,x_n)||\subseteq||\gen(r(x_1,\ldots,x_n))||$. Clearly, $\overline{||\norm(r(x_1,\ldots,x_n),\psi)||}\subseteq||\cogen(r(x_1,\ldots,x_n))||=\V$.
\item Consider the case in which $\varphi$ is of the form $x_1\approx x_2$. Then, clearly,  $||\norm(x_1\approx x_2,\psi)||\subseteq\V=||\gen(x_1\approx x_2)||$ and, similarly, $\overline{||\norm(x_1\approx x_2,\psi)||}\subseteq\V=||\cogen(x_1\approx x_2)||$.

\item Consider the case in which $\varphi$ is of the form $\neg\varphi_1$. Since $||\cogen(\varphi_1)||=||\gen(\neg\varphi_1)||\subseteq||\chi||$, 

We have,  by the induction hypothesis, 
\begin{align*}
||\norm(\neg\varphi_1,\psi)||=\overline{||\norm(\varphi_1,\psi)||}\subseteq||\cogen(\varphi_1)||=||\gen(\neg\varphi_1)||. 
\end{align*}
Similarly, we have 
\begin{align*}
\overline{||\norm(\neg\varphi_1,\psi)||}=\overline{\overline{||\norm(\varphi_1,\psi)||}}=||\norm(\varphi_1,\psi)||\subseteq||\gen(\varphi_1)||=||\cogen(\neg\varphi_1)||. 
\end{align*}
\item Consider the case in which $\varphi$ is of the form $\varphi_1\wedge\varphi_2$. Let $\alpha_1=\norm(\varphi_1,\psi_1)$ and $\alpha_2=\norm(\varphi_2,\psi_2)$, where $\psi_1=(\exists\FV(\psi)\setminus\FV(\varphi_1))\psi$ and $\psi_2=(\exists\FV(\psi)\setminus\FV(\varphi_2))\psi$. We distinguish four subcases.
\begin{enumerate}
\item Consider the subcase in which both $\alpha_1$ and $\alpha_2$ are not negated. We have, by Lemma \ref{lem:norm_eq_wedge}, 
\begin{align*}
||\norm(\varphi_1\wedge\varphi_2,\psi)||&= ||\norm(\varphi_1\wedge\varphi_2,\psi)||\cap||\varphi_1\approx_\wedge\varphi_2||\\
&=||\alpha_1||\cap||\alpha_2||\cap||\varphi_1\approx_\wedge\varphi_2||\\
&\subseteq||\gen(\varphi_1)||\cap||\gen(\varphi_2)||\cap||\varphi_1\approx_\wedge\varphi_2||\\
&=||\gen(\varphi_1\wedge\varphi_2||).
\end{align*}
 
\item Consider the subcase in which $\alpha_1$ is not negated and $\alpha_2$ is negated.
We have
\begin{align}
&||\norm(\varphi_1\wedge\varphi_2,\psi)||=||\al(\alpha_1,\psi)||\cap\overline{||\al(\overline{\alpha_2},\psi)||}\notag\\
\subseteq&||\al(\alpha_1,\psi)||\subseteq||\alpha_1||\subseteq||\gen(\varphi_1)||\notag\\
= &||\gen(\varphi_1)||\cap||\gen(\varphi_2)||\cap||\varphi_1\approx_\wedge\varphi_2||\label{eq:lem_norm1_eq1}\\
=&||\gen(\varphi_1\wedge\varphi_2||).\notag
\end{align}
The equality \eqref{eq:lem_norm1_eq1} follows from Lemmas \ref{lem:pos_neg_coeq_eq} and \ref{lem:pos_neg_cogen_gen}.
 
\item The subcase in which $\alpha_1$ is negative and $\alpha_2$ is positive can be proved analogously to the previous subcase.
\item Consider the subcase in which both $\alpha_1$ and $\alpha_2$ are negated. We have, by  \eqref{eq:align_subseteq} and by the induction hypothesis,
\begin{align*}
||\al(\overline{\alpha_1},\psi)||\subseteq||\overline{\alpha_1}||=\overline{||\alpha_1||}\subseteq||\cogen(\varphi_1)||
\end{align*}
and
\begin{align*}
||\al(\overline{\alpha_2},\psi)||\subseteq||\overline{\alpha_2}||=\overline{||\alpha_2||}\subseteq||\cogen(\varphi_2)||.
\end{align*}
Now, we have
\begin{align*}
\overline{||\norm(\varphi_1\wedge\varphi_2,\psi)||}&=||\al(\overline{\alpha_1},\psi)||\cup||\al(\overline{\alpha_2},\psi)||\\
 &\subseteq||\cogen(\varphi_1)||\cup||\cogen(\varphi_2)||\\
 &\subseteq\C_{X_1\setminus X_2}(||\cogen(\varphi_1)||)\cup\C_{X_2\setminus X_1}(||\cogen(\varphi_2)||)\\
 &=||\cogen(\varphi_1\wedge\varphi_2)||,
\end{align*}
where $X_1=\FV(\cogen(\varphi_1))$ and $X_2=\FV(\cogen(\varphi_2))$. 
\end{enumerate}
\item Consider the case when $\varphi$ is of the form $(\exists x)\varphi_1$. Let $\psi_1=\gen(\varphi_1)$ and $X=\FV(\psi_1)\setminus\{x\}$. From $||\norm(\varphi_1,(\exists X)\psi_1\wedge\psi)||\subseteq||\gen(\varphi_1)||$,
it follows
\begin{align*}
||\norm((\exists x)\varphi_1,\psi)||&=\C_x(||\norm(\varphi_1,(\exists X)\psi_1\wedge\psi)||)\subseteq\C_x(||\gen(\varphi_1)||)\\
&=||\gen((\exists x)\varphi_1)||.
\end{align*}
From the induction hypothesis $||\overline{\norm(\varphi_1,(\exists X)\psi_1\wedge\psi)||}\subseteq||\cogen(\varphi_1)||$, it follows
\begin{align}
\overline{||\norm((\exists x)\varphi_1,\psi)||}&=\overline{\C_x(||\norm(\varphi_1,(\exists X)\psi_1\wedge\psi)||)}\notag\\
&\subseteq\C_x(\overline{||\norm(\varphi_1,(\exists X)\psi_1\wedge\psi)||})\label{eq:lem_norm1_eq2}\\
&\subseteq\C_x(||\cogen(\varphi_1)||)=||\cogen((\exists x)\varphi_1)||\notag.
\end{align}
The inclusion~\eqref{eq:lem_norm1_eq2} follows from item~(3) of Lemma~\ref{lem:Cx_overline}. \qedhere

\end{enumerate} 
\end{proof}
The normalization of an allowed formula $\varphi$ using a formula $\psi$ has the same value as the original formula $\varphi$ within the value of $\psi$.   
\begin{lem}\label{lem:norm2}
Assume that $\varphi$ is allowed and $\FV(\psi)=\FV(\varphi)$. Then it holds
$$||\norm(\varphi,\psi)||\cap||\psi||=||\varphi||\cap||\psi||.$$ 
\end{lem}
\begin{proof} We prove the assertion by structural induction on the formula $\varphi$.
\begin{enumerate}
\item Consider the case when $\varphi$ is of the form $r(x_1,\ldots,x_n)$. Let $y_1,\ldots,y_n$ be variables such that $r(y_1,\ldots,y_n)\in\F$. We have
\begin{align}
 ||\norm(\varphi,\psi)||\cap||\psi||&=||\norm(r(x_1,\ldots,x_n),\psi)||\cap||\psi||\notag\\
 &=||r(y_1/x_1,\ldots,y_n/x_n)||\cap||\psi||\notag\\
 &=||r(x_1,\ldots,x_n)||\cap||\psi||=||\varphi||\cap||\psi||.\label{eq:lem_norm2_2}
\end{align}
The equality \eqref{eq:lem_norm2_2} follows from Lemma \ref{lem:rename}.
\item Consider the case when $\varphi$ is of the form $x\approx y$. Let $X=\FV(\psi)\setminus\{x,y\}$ We have
\begin{align*}
||\norm(\varphi,\psi)||\cap||\psi||&=||(\exists X)\psi\wedge(x\approx y)||\cap||\psi||\\
&=\C_X(||\psi||)\cap||x\approx y||\cap||\psi||\\
&=||x\approx y||\cap||\psi||=||\varphi||\cap||\psi||.
\end{align*}
\item Consider the case when $\varphi$ is of the form $\neg\varphi_1$. We have
\begin{align}
||\norm(\varphi,\psi)||\cap||\psi||&=||\norm(\neg\varphi_1,\psi)||\cap||\psi||\notag\\
&=||\overline{\norm(\varphi_1,\psi)}||\cap||\psi||\notag\\
&=\overline{||\norm(\varphi_1,\psi)||}\cap||\psi||\notag\\
&=\overline{||\norm(\varphi_1,\psi)||\cap||\psi||}\cap||\psi||\notag\\
&=\overline{||\varphi_1||\cap||\psi||}\cap||\psi||=\overline{||\varphi_1||}\cap||\psi||\notag\\
&=||\neg\varphi_1||\cap||\psi||=||\varphi||\cap||\psi||.\notag
\end{align} 
\item Consider the case when $\varphi$ is of the form $\varphi_1\wedge\varphi_2$. Let $\alpha_1=\norm(\varphi_1,\psi)$ and $\alpha_2=\norm(\varphi_2,\psi)$. We distinguish four subcases.
\begin{enumerate}
\item Consider the subcase when both $\alpha_1$ and $\alpha_2$ are positive. We have
\begin{align*}
||\norm(\varphi,\psi)||\cap||\psi||&= ||\norm(\varphi_1\wedge\varphi_2,\psi)||\cap||\psi||\\
&=||\norm(\varphi_1,\psi)\wedge\norm(\varphi_2,\psi)||\cap||\psi||\\
&=||\norm(\varphi_1,\psi)||\cap||\norm(\varphi_2,\psi)||\cap||\psi||\\
&=||\varphi_1||\cap||\varphi_2||\cap||\psi||\\
&=||\varphi_1\wedge\varphi_2||\cap||\psi||=||\varphi||\cap||\psi||.
\end{align*} 
\item Consider the subcase when $\alpha_1$ is positive and $\alpha_2$ is negative. Let $X=\FV(\alpha_1)\cup\FV(\alpha_2)$. We have
\begin{align}
&||\norm(\varphi,\psi)||\cap||\psi||=||\norm(\varphi_1\wedge\varphi_2)||\cap||\psi||\notag\\
&=||\al(\alpha_1,X,\psi)||\cap\overline{||\al(\overline{\alpha_2},X,\psi)||} \cap||\psi||\notag\\
&=||\alpha_1||\cap\overline{||\al(\overline{\alpha_2},X,\psi)||\cap||\psi||} \cap||\psi||\label{eq:lem_norm2_4}\\
&=||\alpha_1||\cap\overline{||\overline{\alpha_2}||\cap||\psi||} \cap||\psi||\label{eq:lem_norm2_5}\\
&=||\alpha_1||\cap\overline{||\overline{\alpha_2}||} \cap||\psi||\notag\\
&=||\alpha_1||\cap\overline{\overline{||\alpha_2||}} \cap||\psi||\notag\\
&=||\alpha_1||\cap||\alpha_2|| \cap||\psi||\notag\\
&=||\varphi_1||\cap||\varphi_2|| \cap||\psi||=||\varphi||\cap||\psi||.\notag
\end{align}
The equalities \eqref{eq:lem_norm2_4} and \eqref{eq:lem_norm2_5} follows from Lemma \ref{lem:align}. 
\item The subcase in which $\alpha_1$ is negative and $\alpha_2$ is positive can be proved analogously to the previous subcase.
\item Consider the subcase where both $\alpha_1$ and $\alpha_2$ are negative.  Let $X=\FV(\alpha_1)\cup\FV(\alpha_2)$.  We have
\begin{align}
&||\norm(\varphi,\psi)||\cap||\psi||=||\norm(\varphi_1\wedge\varphi_2)||\cap||\psi||\notag\\
&=\overline{||\al(\overline{\alpha_1},X,\psi)||\cup||\al(\overline{\alpha_2},X,\psi)||} \cap||\psi||\notag\\
&=\overline{||\al(\overline{\alpha_1},X,\psi)||\cap||\psi||}\cap\overline{||\al(\overline{\alpha_2},X,\psi)||\cap||\psi||} \cap||\psi||\notag\\
&=\overline{||\overline{\alpha_1}||}\cap\overline{||\overline{\alpha_2}||} \cap||\psi||=||\alpha_1||\cap||\alpha_2||\cap||\psi||\notag\\
&=||\varphi_1||\cap||\varphi_2|| \cap||\psi||=||\varphi||\cap||\psi||.\notag
\end{align}
Since this subcase is similar to subcase (b), we can skip some steps. 
\end{enumerate}
  
\item Consider the case when $\varphi$ is of the form $(\exists x)\varphi_1$. Let $\psi_1=\gen(\varphi_1)$ and $X=\FV(\psi_1)\setminus\{x\}$.  We have
\begin{align}
&||\varphi||\cap||\psi||=||(\exists x)\varphi_1||\cap||\psi||=\C_x(||\varphi_1||)\cap||\psi||\notag\\
&=\C_x(||\varphi_1||\cap||\psi_1||)\cap\C_x(||\psi||)\cap||\psi||\label{eq:lem_norm2_9}\\
&=\C_x(||\varphi_1||\cap\C_{X}(||\psi_1||)\cap\C_x(||\psi||))\cap||\psi||\label{eq:lem_norm2_7}\\
&=\C_x(||\varphi_1||\cap\C_{X}(||\psi_1||)\cap||\psi||)\cap||\psi||\label{eq:lem_norm2_11}\\
&=\C_x(||\varphi_1||\cap||(\exists X)\psi_1\wedge\psi||)\cap||\psi||\notag\\
&=\C_x(||\norm(\varphi_1,(\exists X)\psi_1\wedge\psi)||\cap||(\exists X)\psi_1\wedge\psi||)\cap||\psi||\label{eq:lem_norm2_13}\\
&=\C_x(||\norm(\varphi_1,(\exists X)\psi_1\wedge\psi)||\cap\C_X(||\psi_1||)\cap||\psi||)\cap||\psi||\notag\\
&=\C_x(||\norm(\varphi_1,(\exists X)\psi_1\wedge\psi)||\cap\C_X(||\psi_1||)\cap\C_x(||\psi||))\cap||\psi||\label{eq:lem_norm2_12}\\
&=\C_x(||\norm(\varphi_1,(\exists X)\psi_1\wedge\psi)||\cap\C_X(||\psi_1||))\cap\C_x(||\psi||)\cap||\psi||\label{eq:lem_norm2_8}\\
&=\C_x(||\norm(\varphi_1,(\exists X)\psi_1\wedge\psi)||\cap\C_X(||\psi_1||))\cap||\psi||\label{eq:lem_norm2_10}\\
&=\C_x(||\norm(\varphi_1,(\exists X)\psi_1\wedge\psi)||)\cap||\psi||\label{eq:lem_norm2_1}\\
&=||(\exists x)\norm(\varphi_1,(\exists X)\psi_1\wedge\psi)||\cap||\psi||\notag\\
&=||\norm((\exists x)\varphi_1,\psi)||\cap||\psi||\notag\\
&=||\norm(\varphi,\psi)||\cap||\psi||\notag
\end{align}
The equality \eqref{eq:lem_norm2_9} holds due to Lemma \ref{thm:gen}.
The third axiom of cylindric algebras is used in the equalities \eqref{eq:lem_norm2_9} and \eqref{eq:lem_norm2_10}.
The fourth axiom of cylindric algebras is used in the equalities \eqref{eq:lem_norm2_7} and \eqref{eq:lem_norm2_8}.
The equality \eqref{eq:lem_norm2_1}  follows from Lemmas \ref{lem:norm1} and \ref{thm:gen}. The equality \eqref{eq:lem_norm2_13} holds by the induction hypothesis.

Since $\FV(\varphi)=\FV(\psi)$, $x\notin\FV(\psi)$, and thus $\C_x(||\psi||)=||\psi||$. This fact is used in equalities \eqref{eq:lem_norm2_11} and \eqref{eq:lem_norm2_12}. \qedhere
 
\end{enumerate}

\end{proof}

We conclude our study by presenting two main results concerning allowed formulas.

\begin{thm}\label{thm:norm3}
Every allowed formula is semantically equivalent to its normalisation. 
\end{thm}
\begin{proof}
Let $\varphi$ be an allowed formula. Since $\varphi$ is allowed, we have $\FV(\varphi)=\FV(\gen(\varphi))$. By Theorem \ref{thm:gen}, it follows that  $||\varphi||\subseteq||\gen(\varphi)||$, and, by Lemma \ref{lem:norm2}, $||\norm(\varphi,\gen(\varphi))||\cap||\gen(\varphi)||=||\varphi||\cap||\gen(\varphi)||=||\varphi||$.
\end{proof}

\begin{thm}
Allowed formulas are domain independent. 
\end{thm}
\begin{proof}
Let $\varphi$ be an allowed formula. Then the normalization $\norm(\varphi,\gen(\varphi))$ is semantically equivalent to $\varphi$ (Theorem \ref{thm:norm3}) and is in normalized form. By Lemma~\ref{lem:norm_form}, the normalization $\norm(\varphi,\gen(\varphi))$ is semantically equivalent to the relational expression $\expr(\norm(\varphi,\gen(\varphi)))$. By Theorem~\ref{thm:dom_idep}, the normalization is domain independent, and thus $\varphi$ is also domain independent.
\end{proof}

\bibliographystyle{alphaurl}
\bibliography{references}

\end{document}